\documentclass[journal,10pt,final]{IEEEtran}

\usepackage{tikz}
\usetikzlibrary{calc}
\usepackage{cite}
\usepackage{graphicx}
\usepackage{algorithm}
\usepackage{algorithmic}
\usepackage{epsfig}
\usepackage{epstopdf}
\usepackage{amssymb,amsmath}
\usepackage{amsmath}
\usepackage{amsthm}
\usepackage{amsfonts}
\usepackage{subfigure}
\usepackage{psfrag}
\usepackage{rotating}
\usepackage{latexsym}
\usepackage{dsfont}
\usepackage{stmaryrd}
\usepackage{amsthm}
\usepackage{amssymb}
\usepackage{amsmath}

\usepackage[force,almostfull]{textcomp}
\usepackage{fancyvrb}
\usepackage{textcomp}
\usepackage{url}
\usepackage{ifdraft}
\usepackage{url}
\usepackage{multirow}
\usepackage{rotating}
\usepackage{xspace}
\usepackage{array}
\usepackage{arydshln}
\usepackage{multido}
\usepackage{flushend}
 \usepackage{enumitem}
\usepackage{enumerate}
\usepackage[latin1]{inputenc}
\newlength\figureheight 
\newlength\figurewidth 
\usepackage{graphics}
\usepackage{pgfplots}
\pgfplotsset{plot coordinates/math parser=false}

\usepackage{scalefnt}

\newtheorem{specialcasecounter}{Theorem}
\newtheoremstyle{specialcasestyle}{1mm}{1mm}{\upshape}{}{\bfseries\upshape}{.}{0mm}{}
\theoremstyle{specialcasestyle}

\newtheorem{rem}{Remark}

\begin{document}

\title{An Improved Hazard Rate Twisting Approach for the Statistic of the Sum of Subexponential Variates (Extended Version)}

\author{Nadhir~Ben Rached, Abla Kammoun, Mohamed-Slim~Alouini, and Raul Tempone
\\\thanks { The authors are with CEMSE Division, King Abdullah University of Science and Technology (KAUST), Thuwal, Makkah Province, Saudi Arabia. E-mails: \{nadhir.benrached, abla.kammoun, slim.alouini, raul.tempone\}@kaust.edu.sa.

The authors are members of the KAUST Strategic Research Initiative on Uncertainty Quantification in Science and Engineering (SRI-UQ)}\vspace{-0.8cm}
}
\date{}
\maketitle
\thispagestyle{empty}
\begin{abstract}
In this letter, we present an improved hazard rate twisting technique for the estimation of the probability that a sum of independent but not necessarily identically distributed subexponential Random Variables (RVs) exceeds a given threshold. Instead of twisting all the components in the summation, we propose to twist only the RVs which have the biggest impact on the right-tail of the sum distribution and keep the other RVs unchanged. A minmax approach is performed to determine the optimal twisting parameter which leads to an asymptotic optimality criterion. Moreover, we show through some selected simulation results that our proposed approach results in a variance reduction compared to the technique where all the components are twisted.
\end{abstract}

\vspace{-3mm}
\section{Introduction}
{\ The problem of analyzing the statistic of the sum of RVs is often encountered in many applications of wireless communication systems. The most relevant example is that of estimating the probability that the total interference power \cite{930627}, often modeled  as a sum of RVs, exceeds a certain threshold. Addressing this issue can essentially serve to shed the light on the behavior of the outage probability of the signal-to-interference-plus-noise (SINR) ratio, which is among the most important performance metrics in practice. This question is receiving an increasing interest,  mostly spurred by the emergence of cognitive radio systems, in which the control of the interference is of paramount importance \cite{6364048}. However, closed-form expressions for many well-known challenging sum distributions are generally intractable and unknown, which makes this problem far from being trivial. In particular, we are interested in estimating the sum distribution of RVs with subexponential decay which includes for example the Log-normal and the Weibull (with shape parameter less than $1$) RVs. The sum distribution of these two RVs has received a lot of interests \cite{1275712,1369233,4814351,4623761,5425708,5282371,1665128,alouini1,1388722}. A crude Monte Carlo (MC) simulation is of course the standard technique to estimate the probability of interest. However, it is widely known that this technique becomes computationally expensive when  rare events are considered, i.e. events with very small probabilities.\par }
{\ Importance Sampling (IS) is a well-known variance reduction technique which aims to improve the computational efficiency of naive MC simulations \cite{opac-b1100693}. The idea behind this technique is to consider a suitable change of the underlying sampling distribution in a way to achieve a substantial variance reduction of the IS estimator. Many research efforts were carried out to propose efficient IS simulation approaches. For instance, an interesting hazard rate twisting technique was derived that deals with the sum of independent and identically distributed (i.i.d) subexponential RVs \cite{Juneja:2002:SHT:566392.566394}. In \cite{DBLP:journals/corr/RachedBKAT14}, this technique was further extended to handle the sum of independent and not necessarily identically distributed subexponential RVs. The idea behind this technique was to twist the hazard rate of each component in the summation by a twisting parameter $\theta$ between $0$ and $1$. In this letter, we propose an improved version of the work in \cite{DBLP:journals/corr/RachedBKAT14} by twisting only the heaviest components which have the biggest impact on the right-tail of the sum distribution. \par }
\vspace{-3mm}
\section{Importance Sampling}
{\ Let us consider a sequence $X_1,X_2,...,X_N$ of independent and not necessarily identically distributed positive RVs. We denote the Probability Density Function (PDF) of $X_i$ by $f_i(\cdot)$, $i=1,2,...,N$. Our goal is to efficiently estimate: 
\begin{align}
\alpha=P \left ( S_N=\sum_{i=1}^{N}{X_i} >\gamma_{th}\right ),
\end{align}
for a sufficiently large threshold $\gamma_{th}$. We focus on the case where the RVs $X_i$, $i=1,2...,N$, belong to the class of heavy-tailed distributions, i.e distributions which decay slower than any exponential distribution. In particular, we are interested in a subclass called subexponential distributions which contains some of the most commonly used heavy-tailed distributions such as the Log-normal distribution and the Weibull distribution with shape parameter less than $1$.  
\par }
{\ Obviously, a naive MC simulation technique can be used to estimate $\alpha$. In the framework of small threshold values, this naive MC simulation is actually efficient and no further computational improvement is worthy to try. However, it is well known that this approach is computationally expensive when we consider the estimation of rare events, i.e. events with very small probabilities. IS is an alternative approach which can improve the computational efficiency of naive MC simulations \cite{opac-b1100693}. The idea behind this variance reduction technique is to construct an unbiased estimator $\hat \alpha_{IS}$ of $\alpha$ with much smaller variance than the naive MC estimator. In fact, the IS approach is based on performing a change of the sampling distribution as follows:
\begin{align}
\nonumber \alpha &=\mathbb{E}\left [ \textbf{1}_{(S_N>\gamma_{th})}\right ]=\int_{\mathbb{R}^N}{\textbf{1}_{(S_N>\gamma_{th})}\prod_{i=1}^{N}{f_i(x_i)} dx_1dx_2...dx_N}\\
\nonumber  &= \int_{\mathbb{R}^N}{\textbf{1}_{(S_N>\gamma_{th})} L(x_1,x_2,...,x_N)\prod_{i=1}^{N}{g_i(x_i)}dx_1dx_2...dx_N}\\
&=\mathbb{E}_{p^*} \left [ \textbf{1}_{(S_N>\gamma_{th})} L(X_1,X_2,...,X_N)\right ],
\end{align}
where the expectation $\mathbb{E}_{p^*}\left [ \cdot \right ]$ is taken with respect to the new probability measure $p^*$ under which the PDF of each $X_i$ is $g_i(\cdot)$, $i=1,2,...,N$. The likelihood ratio $L$ is defined as:
\begin{align}\label{like}
L(X_1,X_2,...,X_N)=\prod_{i=1}^{N}{\frac{f_i(X_i)}{g_i(X_i)}}.
\end{align}
The IS estimator is then given by: 
\begin{align}
\hat \alpha_{IS}=\frac{1}{M}\sum_{i=1}^{M}{\textbf{1}_{(S_N(\omega_i)>\gamma_{th})}L(X_1(\omega_i),...,X_N(\omega_i))},
\end{align}
where $M$ is the number of simulation runs and $\textbf{1}_{(\cdot)}$ is the indicator function. The fundamental issue in importance sampling lies in the choice of the new sampling distribution $g_i(\cdot)$, $i=1,2,...,N$, that results in  substantial computational gains. In fact, the new sampling distribution should emphasize the generation of values that have more impact on the desired probability $\alpha$ (in our setting, important samples are the ones which satisfy $S_N>\gamma_{th}$). Thus, by sampling these important values frequently, the estimator's variance can be reduced. 
\par }
{\ The asymptotic optimality is an interesting criterion that can be used to quantify the pertinence of the probability measure change \cite{Juneja:2002:SHT:566392.566394}. This criterion is often achieved through a clever choice of the sampling distribution. Let us define the sequence of RVs $T_{\gamma_{th}}$ as:
\begin{align}\label{Tgamma}
T_{\gamma_{th}}=\textbf{1}_{(S_N>\gamma_{th})}L(X_1,X_2,...,X_N).
\end{align}
From the non-negativity of the variance of $T_{\gamma_{th}}$, it follows that:
\begin{align}\label{res1}
\frac{\log \left (\mathbb{E}_{p^*}\left [ T_{\gamma_{th}}^2 \right ]\right )}{\log\left (\alpha \right )} \leq 2,
\end{align}
for any probability measure $p^*$. We say that the asymptotic optimality criterion is achieved if the previous inequality holds with equality as $\gamma_{th} \rightarrow +\infty$, that is:
\begin{align}\label{asymp_opt}
\lim_{\gamma_{th} \rightarrow +\infty}{\frac{\log \left (\mathbb{E}_{p^*}\left [ T_{\gamma_{th}}^2 \right ]\right )}{\log \left (\alpha \right )} }=2.
\end{align}
For instance, the naive MC simulation is not asymptotically optimal since the limit in (\ref{asymp_opt}) is equal to $1$. A lot of research efforts were carried out to propose interesting changes of the sampling distribution. The exponential twisting technique, derived from the large deviation theory, is the most used technique in the setting where the underlying distribution is light-tailed \cite{54903}. In the heavy-tailed case, the exponential twisting technique is no more feasible and an alternative approach is needed. In \cite{Juneja:2002:SHT:566392.566394}, an interesting hazard rate twisting technique (we call it conventional hazard rate-based IS technique in the rest of this work) was derived to deal with heavy-tailed distributions. In that work, an i.i.d sum of distributions with subexponential decay were considered and the asymptotic optimality of the hazard rate twisting approach was verified. In \cite{DBLP:journals/corr/RachedBKAT14}, an extension of \cite{Juneja:2002:SHT:566392.566394} to the case of independent and not necessarily identically distributed sum of subexponential variates is developed and the asymptotic optimality criterion was again shown to be satisfied. Let us define the hazard rate of $X_i$, $i=1,2,...,N$, as
\begin{align}
\lambda_i(x)=\frac{f_i(x)}{1-F_i(x)}, \text{ }x>0,
\end{align}
where $F_i(\cdot)$ is the cumulative distribution function of $X_i$, $i=1,2,...,N$. We also define  the hazard function as
\begin{align}
\nonumber \Lambda_i(x)&=\int_{0}^{x}{\lambda_i(t)dt}=-\log(1-F_i(x)), \text{ }x>0.
\end{align}
The PDF of $X_i$, $i=1,2,...,N$, is related to the hazard rate and the hazard function as follows
\begin{align}\label{under}
f_i(x)=\lambda_i \left ( x \right )\exp \left ( -\Lambda_i(x)\right ), \text{ }x>0.
\end{align}
The conventional hazard rate twisting technique \cite{DBLP:journals/corr/RachedBKAT14} considers a new sampling distribution which is obtained by twisting the hazard rate of the underlying distribution (\ref{under}) of each component in the summation $S_N$ by the same quantity $0 \leq\theta<1$
\begin{align}
\label{eq:twist}
g_i(x)\triangleq f_{i,\theta}(x)=(1-\theta)\lambda_i(x)\exp \left (-(1-\theta)\Lambda_i(x) \right ).
\end{align} 
\par }
\section{Improved Approach}
\subsection{Proposed Approach}
{\ Similar to \cite{DBLP:journals/corr/RachedBKAT14}, we consider a sequence $X_1,X_2,...,X_N$ of independent and not necessarily identically distributed subexponential positive RVs, belonging to the same family of distribution, i.e., for example a sum of independent Log-normal variates with different means and variances. From this sequence, we extract a sub-sequence containing the RVs with the heaviest right-tail which are naturally i.i.d RVs. Let us denote by $s$ the number of RVs contained in this sub-sequence. It is important to note that the particular i.i.d case occurs when $s=N$. For instance, for the particular Weibull distributions with scale parameters $\beta_i$ and shape parameters $k_i$, $i=1,2,...,N$, the number $s$ is defined as:
\begin{align}
s&=\#\{i\in \{1,2,...,N\} \text{ such that } k_i=k_{min} \text{ and }\beta_i=\beta_{max} \},
\end{align}
where $\#$ denotes the cardinality of the set, $k_{min}=\min_{i}{k_i}$ and $\beta_{max}=\max_{i;k_i=k_{min}}{\beta_i}$. For Log-normal RVs with mean $\mu_i$ and standard deviation $\sigma_i$, $i=1,2,...,N$, the number $s$ is 
\begin{align}
s&=\#\{i\in \{1,2,...,N\} \text{ such that } \sigma_i=\sigma_{max} \text{ and }\mu_i=\mu_{max} \},
\end{align}
where $\sigma_{max}=\max_{i}{\sigma_i}$, and $\mu_{max}=\max_{i;\sigma=\sigma_{max}}{\mu_i}$. The adopted methodology in the present work is to twist only these $s$ heaviest i.i.d RVs and keep the others untwisted. The intuition behind this methodology is that the remaining untwisted RVs have a negligible effect on the right-tail of the sum distribution. Moreover, by only twisting the dominating RVs, the new sum distribution is less heavier than the one obtained by the conventional approach (since the hazard rate twisting of a RV results in a more heavier distribution). Furthermore, our improved technique remains able to generate important samples, i.e realizations that exceed the given threshold. Therefore, one could expect a variance reduction using our improved approach. In order to validate  our expectation, we represent in Fig. \ref{fig1} the second moment of the RV $T_{\gamma_{th}}$, given by the conventional and the improved IS approaches, as function of the twisting parameter $\theta$ for the sum of four independent but not identically distributed Log-normal RVs and for a fixed threshold $\gamma_{th}$.
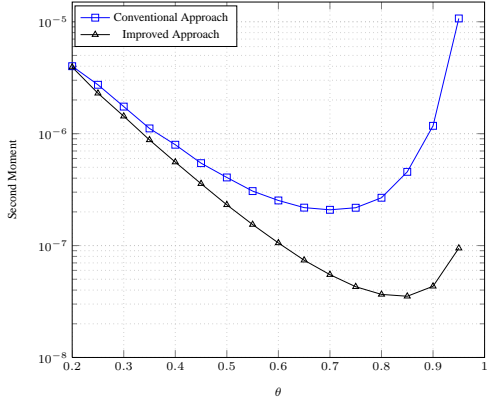
\begin{figure}[h]
\centering
\scalefont{0.7}
\begin{tikzpicture}[scale=0.65]
\begin{semilogyaxis}[%
scale only axis,
xmin=0.2, xmax=1,
xlabel={$\theta$},
xmajorgrids,
ymin=1e-08, ymax=0.000015,
yminorticks=true,
ymajorgrids,
yminorgrids,
grid style={dotted},
ylabel={Second Moment},
legend style={at={(0.4,0.990132389187107)},draw=black,fill=white,align=left}]
\addplot [
color=blue,
solid,
mark=square,
mark options={solid}
]
coordinates{
 (0.2,4.00634477079729e-06)(0.25,2.73367095817808e-06)(0.3,1.744904803475e-06)(0.35,1.11350404145892e-06)(0.4,7.96200453494425e-07)(0.45,5.44881423679326e-07)(0.5,4.06050730679034e-07)(0.55,3.06682140054296e-07)(0.6,2.52904362741175e-07)(0.65,2.18250220097286e-07)(0.7,2.08703479000034e-07)(0.75,2.17744487149853e-07)(0.8,2.67390005646349e-07)(0.85,4.55738409657214e-07)(0.9,1.1744505045028e-06)(0.95,1.07019373364102e-05) 
};
\addlegendentry{Conventional Approach};

\addplot [
color=black,
solid,
mark=triangle,
mark options={solid}
]
coordinates{
 (0.2,3.93065269030992e-06)(0.25,2.2960675299679e-06)(0.3,1.43162351180022e-06)(0.35,8.79204734079206e-07)(0.4,5.57076245644658e-07)(0.45,3.5707586824853e-07)(0.5,2.3120224916498e-07)(0.55,1.54134159567812e-07)(0.6,1.05524661377818e-07)(0.65,7.39591084949413e-08)(0.7,5.49636205211449e-08)(0.75,4.28663169779973e-08)(0.8,3.66108875131944e-08)(0.85,3.52413839666486e-08)(0.9,4.32706471196361e-08)(0.95,9.47350162754793e-08) 
};
\addlegendentry{Improved Approach};

\end{semilogyaxis}
\end{tikzpicture}%
\vspace{-3mm}
\caption{Second moment of $T_{\gamma_{th}}$ for the sum of four independent Log-normal RVs with mean $\mu_i=0\text{ dB},i=1,2,3,4$, standard deviations $\sigma_1=\sigma_2=4$ dB, $\sigma_3=\sigma_4=6$ dB, and $\gamma_{th}=25$ dB.}
\label{fig1}
\end{figure}
It is clear from this figure that the idea of considering only the dominating RVs, instead of treating all the components similarly, reduces the second moment $\mathbb{E}_{p^*} \left [T_{\gamma_{th}}^2 \right]$ for all values of $\theta$ and hence decreases the variance of $T_{\gamma_{th}}$. In the following subsection, we will describe a procedure to determine the optimal (in a sense that will be explained later) twisting parameter which ensures the largest amount of variance reduction. Without loss of generality, we assume that the $s$ heaviest i.i.d RVs which will be twisted are $X_1,X_2,...,X_s$. From \eqref{like} and \eqref{eq:twist}, the likelihood ratio of the improved approach is then:
\begin{align}\label{Like_ratio}
L(X_1,X_2,...,X_s)=\frac{1}{\left (1-\theta \right )^s}\exp \left (  -\theta \sum_{i=1}^{s}{\Lambda_1(X_i)}\right ),
\end{align}
where $\Lambda_1(\cdot)$ is the hazard function of the RVs $X_1,X_2,...,X_s$.
\par }
\vspace{-3mm}
\subsection{Determination of the Twisting Parameter}
{\ The determination of the parameter $\theta$ is performed following the steps of the minmax approach derived in \cite{DBLP:journals/corr/RachedBKAT14}. It is important to note that applying the same technique in the present setting is not attractive since it results in a zero twisting parameter. Hence, a slight modification is required. In fact, the second moment of $T_{\gamma_{th}}$ can be decomposed into two terms as follows:
\begin{align}
&\nonumber \mathbb{E}_{p^*} \left [T_{\gamma_{th}}^2 \right ]=\mathbb{E}_{p^*} \left [ L^2(X_1,X_2,...,X_s)\textbf{1}_{(S_s>\gamma_{th})} \right ]\\
&+\mathbb{E}_{p^*} \left [ L^2(X_1,X_2,...,X_s)\textbf{1}_{(S_N>\gamma_{th},S_s<\gamma_{th})} \right ].
\end{align}
Instead of applying the minmax approach to $\mathbb{E}_{p^*} \left [T_{\gamma_{th}}^2 \right ]$ as in \cite{DBLP:journals/corr/RachedBKAT14}, we propose to determine the parameter $\theta$ through considering only the dominant term $\mathbb{E}_{p^*} \left [ L^2(X_1,X_2,...,X_s)\textbf{1}_{(S_s>\gamma_{th})} \right ]$. In the first step, an upper bound of this term is derived through the resolution of the following maximization problem $(P)$
\begin{align}
\underset{X_1,...,X_s}{\max} \hspace{3mm}&L(X_1,X_2,...,X_s)\\
\text{Subject to   } &\sum_{i=1}^{s}{X_i} \geq\gamma_{th}, \text{    }\nonumber X_i>0, \hspace{2mm} i=1,...,s.
\end{align}
Let $X_1^*,X_2^*,...,X_s^*$ be the solution of $(P)$. From (\ref{Like_ratio}), it follows that
\begin{align}
&\nonumber \mathbb{E}_{p^*} \left [ L^2(X_1,X_2,...,X_s)\textbf{1}_{(S_s>\gamma_{th})} \right ]\\
\label{eq2} & \leq\frac{1}{\left (1-\theta \right )^{2s}}\exp\left (-2\theta\sum_{i=1}^{s}{\Lambda_1(X_i^*)} \right )=h\left (\theta \right). 
\end{align}
The second step in the minmax approach is to minimize the previous upper bound with respect to $\theta$. Equivalently, we minimize the function $\log \left (h \left (\theta \right ) \right )$. Equating The first derivative of $\log(f(\theta))$ with respect to $\theta$ to zero yields
\begin{align*}
-2 \sum_{i=1}^{s}{\Lambda_1\left( X_i^*\right)}+\frac{2s}{1-\theta^*}=0.
\end{align*}
This leads to the minmax optimal parameter given by
\begin{align}\label{theta_opti}
\theta^*=1-\frac{s}{\sum_{i=1}^{s}{\Lambda_1(X_i^*)}}. 
\end{align}
Through a simple computation, we prove that the function $\log \left ( h\left ( \theta \right ) \right)  $ is actually convex and hence $\theta^*$ is a minimizer.
\par }
\subsection{Asymptotic Optimality}
The asymptotic optimality criterion (\ref{asymp_opt}) of our proposed IS technique is stated in the the following theorem:
\begin{specialcasecounter}
For a sum of independent subexponential variates, the quantity $\alpha$ is asymptotically optimally estimated using the improved hazard rate twisting approach with the twisting parameter given in (\ref{theta_opti}) and provided that $P \left (X_i>\gamma_{th} \right )=o(P \left ( X_1>\gamma_{th}\right )^2)$, $i=s+1,...,N$, as $\gamma_{th} \rightarrow +\infty$..
\end{specialcasecounter}
\begin{proof}
{\ The second moment of the RV $T_{\gamma_{th}}$ could be written using (\ref{Like_ratio}) as follows
\begin{align}
\nonumber &\mathbb{E}_{p^*} \left [ T_{\gamma_{th}}^2\right ]\\
\nonumber &=\int_{S_N >\gamma_{th}}{L^2(x_1,x_2,...,x_s) \prod_{i=s+1}^{N}{f_i(x_i)}\prod_{i=1}^{s}{g_1(x_i)}dx_1dx_2...dx_N}\\
\nonumber &=\int_{S_N >\gamma_{th}}{\frac{1}{(1-\theta)^{N+s}} \exp \left ( -\theta (2 \sum_{i=1}^{s}{\Lambda_1(x_i)} +\sum_{i=s+1}^{N}{\Lambda_i(x_i)})\right )}.\\
\label{secon_mom}&\prod_{i=s+1}^{N}{g_i(x_i)}\prod_{i=1}^{s}{g_1(x_i)}dx_1dx_2...dx_N
\end{align}
Let us now consider the following minimization problem (P')
\begin{align}
\underset{X_1,...,X_N}{\min} \hspace{3mm}& 2 \sum_{i=1}^{s}{\Lambda_1(X_i)} +\sum_{i=s+1}^{N}{\Lambda_i(X_i)}\\
\text{Subject to   } &\sum_{i=1}^{N}{X_i} \geq\gamma_{th}, \text{    }\nonumber X_i>0, \hspace{2mm} i=1,...,N.
\end{align}
Let us denote by $X_1^{*'},X_2^{*'},...,X_N^{*'}$ the solution of (P') and $A'(\gamma_{th})=2 \sum_{i=1}^{s}{\Lambda_1(X_i^{*'})} +\sum_{i=s+1}^{N}{\Lambda_i(X_i^{*'})}$. From (\ref{secon_mom}), we have
\begin{align}
\mathbb{E}_{p^*} \left [ T_{\gamma_{th}}^2\right ] \leq \frac{1}{(1-\theta)^{N+s}}\exp \left (-\theta A'(\gamma_{th}) \right ).
\end{align}
By replacing $\theta$ by $\theta^*$ in (\ref{theta_opti}), it follows that
\begin{align}
\mathbb{E}_{p^*} \left [ T_{\gamma_{th}}^2\right ] \leq \left ( \frac{A(\gamma_{th})}{s}\right )^{N+s} \exp \left ( -A'(\gamma_{th})+s\frac{A'(\gamma_{th})}{A(\gamma_{th})}\right ),
\end{align}
where $A(\gamma_{th})=\sum_{i=1}^{s}{\Lambda_1(X_i^*)}$. By applying the Logarithm function on both sides, we get
\begin{align}\label{secon_log}
\log \left ( \mathbb{E}_{p^*} \left [ T_{\gamma_{th}}^2\right ]\right ) \leq  \left ( N+s\right ) \log ( \frac{A(\gamma_{th})}{s})-A'(\gamma_{th})+s\frac{A'(\gamma_{th})}{A(\gamma_{th})}
\end{align}
In the other hand, since $\{ X_1 >\gamma_{th}\} \subset \{ S_N>\gamma_{th}\}$ (this follows from the positivity of $X_1,X_2,...,X_N$), we get by applying the Logarithm function that
\begin{align}\label{alpha_bn}
\log \left (P(X_1>\gamma_{th}) \right )=-\Lambda_1(\gamma_{th}) \leq \log \left ( \alpha \right ).
\end{align}
The last step of the proof is to investigate the asymptotic behavior of the optimization problems (P) and (P'). Under a concavity assumption (which is satisfied by all commonly used subexponential distributions such as the Log-normal RV and the Weibull RV with shape parameter less than 1), an equivalent optimization problem was studied in details in \cite{DBLP:journals/corr/RachedBKAT14}. Applying the results of \cite{DBLP:journals/corr/RachedBKAT14} to (P) yields
\begin{align}\label{beh1}
A(\gamma_{th}) \sim \Lambda_1(\gamma_{th}), \text{  as  }\gamma_{th} \rightarrow +\infty.
\end{align}
Moreover, the assumption $P \left (X_i>\gamma_{th} \right )=o(P \left ( X_1>\gamma_{th}\right )^2)$, $i=s+1,...,N$, as $\gamma_{th} \rightarrow +\infty$ implies that $2\Lambda_1(\gamma_{th}) -\Lambda_i(\gamma_{th}) \rightarrow -\infty$, $i=s+1,...,N$, as $\gamma_{th} \rightarrow +\infty$. Thus, the results in \cite{DBLP:journals/corr/RachedBKAT14} applied to (P') results in
\begin{align}\label{beh2}
A'(\gamma_{th}) \sim 2\Lambda_1(\gamma_{th}), \text{  as  } \gamma_{th} \rightarrow +\infty.
\end{align}
From (\ref{beh1}), (\ref{beh2}) and the fact that $\Lambda_1(\gamma_{th}) \rightarrow +\infty$  as $\gamma_{th} \rightarrow +\infty$, we deduce that the left hand side of (\ref{secon_log}) is negative for a sufficiently large $\gamma_{th}$. Hence, since $\log \left ( \alpha\right )<0$, it follows that
\begin{align}
\frac{\log \left ( \mathbb{E}_{p^*} \left [ T_{\gamma_{th}}^2\right ]\right )}{\log \left ( \alpha\right )} \geq \frac{\left ( N+s\right ) \log ( \frac{A(\gamma_{th})}{s})-A'(\gamma_{th})+s\frac{A'(\gamma_{th})}{A(\gamma_{th})}}{-\Lambda_1(\gamma_{th})}.
\end{align} 
Finally, using (\ref{beh1}) and (\ref{beh2}), we deduce that
\begin{align}
\lim_{\gamma_{th}\rightarrow +\infty}{\frac{\log \left ( \mathbb{E}_{p^*} \left [ T_{\gamma_{th}}^2\right ]\right )}{\log \left ( \alpha\right )}} \geq 2.
\end{align}
Using (\ref{res1}), the asymptotic optimality criterion (\ref{asymp_opt}) is satisfied and the proof is concluded.
\par }
\end{proof}
\begin{rem}
{\ The assumption that $P \left (X_i>\gamma_{th} \right )=o(P \left ( X_1>\gamma_{th}\right )^2)$, $i=s+1,...,N$, as $\gamma_{th} \rightarrow +\infty$ does not introduce in most of the cases a strong limitation. For instance, this assumption holds for the Weibull distribution when the source of heaviness is due to the shape parameter , that is $k_1< k_i$, $i=s+1,...,N$. Moreover, it is also satisfied in the Log-normal setting provided that $\sigma_1> \sqrt{2}\sigma_i$, $i=s+1,...,N$.
\par}
\end{rem}
\begin{rem}
{\ Through extensive simulation results, we believe that Theorem 1 holds also when the assumption $P \left (X_i>\gamma_{th} \right )=o(P \left ( X_1>\gamma_{th}\right )^2)$, $i=s+1,...,N$, is not satisfied.
\par}
\end{rem}
\section{Simulation Results and Concluding Remarks}
{\ In this section, we will present some selected simulation results to show the computational gain achieved by our new proposed approach compared to the naive Monte Carlo simulation technique and the conventional hazard rate-based technique \cite{DBLP:journals/corr/RachedBKAT14} where all the components in the summation are twisted. In Fig. \ref{fig2}, our objective is to quantify the amount of variance reduction reached by our proposed IS technique. 
\begin{figure}[h]
\centering
\scalefont{0.6}
\begin{tikzpicture}[scale=0.8]
\usetikzlibrary{arrows}
\tikzset{>=latex}
\node[draw] at (2.5,1.77) {$N=2$};
\node[draw] at (4.86,0.62) {$N=4$};
\begin{semilogyaxis}[%
xmin=20, xmax=32,
xlabel={$\gamma{}_{\text{th}}\text{(dB)}$},
xmajorgrids,
ymin=1e-22, ymax=0.01,
yminorticks=true,
ymajorgrids,
yminorgrids,
grid style={dotted},
ylabel={Second Moment},
legend style={at={(0.988703125,0.9879805231856387)},draw=black,fill=white,align=left}]
\addplot [
color=blue,
solid,
mark=square,
mark options={solid}
]
coordinates{
 (20,8.83224792649946e-05)(21,3.03134284563736e-05)(22,9.26820055726201e-06)(23,2.49716015775874e-06)(24,5.92205760348811e-07)(25,1.19555009352404e-07)(26,2.03155592621197e-08)(27,2.88164246857723e-09)(28,3.35881220080318e-10)(29,3.11552282754936e-11)(30,2.2713573534653e-12)(31,1.27467905173407e-13)(32,5.26170929247616e-15) 
};
\addlegendentry{Conventional Approach};

\addplot [
color=black,
solid,
mark=triangle,
mark options={solid}
]
coordinates{
 (20,3.25620707829999e-05)(21,1.03931615434774e-05)(22,2.96352739603601e-06)(23,7.4360607187343e-07)(24,1.62203453023941e-07)(25,3.03275215576415e-08)(26,4.79048223463138e-09)(27,6.26449639495716e-10)(28,6.72227328117003e-11)(29,5.76772430680496e-12)(30,3.87179587553059e-13)(31,1.98047921413902e-14)(32,7.58364273373833e-16)
};
\addlegendentry{Improved Approach};

\addplot [
color=blue,
solid,
mark=square,
mark options={solid},
forget plot
]
coordinates{
 (20,0.000349239341130087)(21,0.000126304037127353)(22,4.13978384409418e-05)(23,1.1963377837103e-05)(24,3.10676858027193e-06)(25,6.88186726902928e-07)(26,1.31963568994924e-07)(27,2.13010346266433e-08)(28,2.8984721053726e-09)(29,3.05260822819919e-10)(30,2.46613106185547e-11)(31,1.53885777775148e-12)(32,7.85215424090042e-14)
};
\addplot [
color=black,
solid,
mark=triangle,
mark options={solid},
forget plot
]
coordinates{
 (20,3.6308140609007e-05)(21,1.14856199469778e-05)(22,3.2351568186305e-06)(23,8.00540439782174e-07)(24,1.73261170262844e-07)(25,3.21156006715677e-08)(26,5.02796282718925e-09)(27,6.57269694798367e-10)(28,6.98274262221075e-11)(29,5.93667928206716e-12)(30,3.98721512686826e-13)(31,2.03315537048712e-14)(32,7.75642978175343e-16)
};
\draw[thick,->](axis cs: 24.5,1e-15)--(axis cs: 26.5,0.9e-8);
\draw[thick,->](axis cs: 24.5,1e-15)--(axis cs: 25.5,1.09e-8);
\draw[thick,->](axis cs: 28.5,1e-19)--(axis cs: 30.5,1e-13);
\draw[thick,->](axis cs: 28.5,1e-19)--(axis cs: 29.5,1.0e-10);
\end{semilogyaxis}
\end{tikzpicture}
\vspace{-3mm}
\caption{Second moment of $T_{\gamma_{th}}$ for the sum of $N$ independent Weibull RVs. First case $N=2$: the scale parameters are $\beta_i=1,i=1,2.$ and the shape parameters are $k_1=0.4$ and $k_2=0.8$. Second case $N=4$: the scale parameters are $\beta_i=1,i=1,2,3,4$, and the shape parameters are $k_1=0.4$ and $k_2=k_3=k_4=0.8$.}
\label{fig2}
\end{figure}
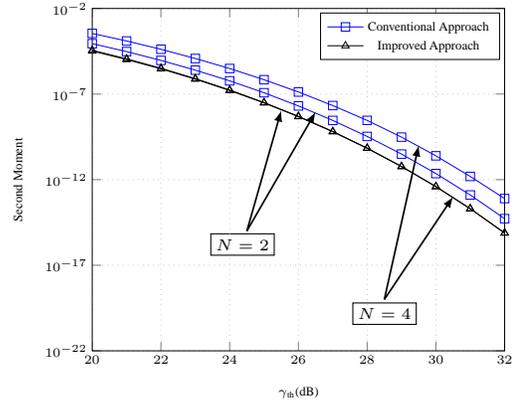
For this purpose, we plot in this figure the second moment of $T_{\gamma_{th}}$ given by both the improved and the conventional techniques as function of the threshold. The minmax optimal parameter (\ref{theta_opti}) is used for the improved IS technique while for the conventional IS technique the corresponding minmax parameter is given in \cite{DBLP:journals/corr/RachedBKAT14}. We deduce from Fig. \ref{fig2} that, for both cases $N=2$ and $N=4$, the variance is  reduced for all the range of considered thresholds.
Moreover, it is worthy to point out that the  variance reduction of the improved method is much more important when $N=4$ than when $N=2$.  
This result is expected since in the case of $N=4$, the two additional components are lighter than the dominating RV and hence they do not affect considerably the right-tail of the sum distribution. Consequently, twisting these two RVs by the conventional IS approach will obviously worsen the second moment of $T_{\gamma_{th}}$. We also deduce from Fig. \ref{fig2} that for $N=4$, the second moment obtained by the improved IS approach remains approximately the same as for $N=2$. The argument is that the variance given by the improved IS technique depends strongly on the number of dominating RVs. Hence, adding the two non-dominating RVs, which are not twisted, to the sum does not affect considerably the variance of the proposed IS method. 
\par}
{\ In a second step, we evaluate the computational gain of the proposed method. For that, we define, for a fixed accuracy requirement,  the efficiency measure $\xi_1$ between the improved IS technique and the naive MC simulation approach as:
\begin{align}
\xi_1=\frac{M_{MC}}{M_I}=\frac{\alpha \left (1-\alpha \right )}{\mathrm{var}_{I} \left [ T_{\gamma_{th}}\right ]} 
\end{align}
Similarly, we define the efficiency $\xi_2$ between the conventional IS technique and the naive MC simulation technique as: 
\begin{align}
\xi_2=\frac{M_{MC}}{M_C}=\frac{\alpha \left (1-\alpha \right )}{\mathrm{var}_{C} \left [ T_{\gamma_{th}}\right ]} 
\end{align}
where $M_{MC}$, $M_I$, $M_C$, are the number of simulation runs for respectively, the naive MC simulation technique, the improved IS and the conventional  IS techniques, whereas $\alpha \left (1-\alpha \right )$, $\mathrm{var}_{I} \left [ T_{\gamma_{th}}\right ]$ and $\mathrm{var}_{C} \left [ T_{\gamma_{th}}\right ]$ refer to their corresponding variances. 
The efficiency $\xi_1$ (respectively $\xi_2$) measures the gain achieved by the improved IS technique (respectively the conventional IS technique) over the naive MC simulation technique in terms of necessary number of simulation runs to meet a fixed accuracy requirement. 
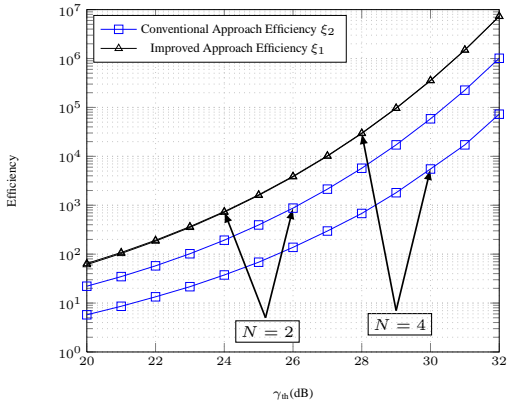
\begin{figure}[h]
\centering
\scalefont{0.6}
\begin{tikzpicture}[scale=0.8]
\usetikzlibrary{arrows}
\tikzset{>=latex}
\node[draw] at (3.0,0.34) {$N=2$};
\node[draw] at (5.2,0.46) {$N=4$};
\begin{semilogyaxis}[%
xmin=20, xmax=32,
xlabel={$\gamma{}_{\text{th}}\text{(dB)}$},
xmajorgrids,
ymin=1, ymax=10000000,
yminorticks=true,
ylabel={Efficiency},
ymajorgrids,
yminorgrids,
grid style={dotted},
legend style={at={(0.015872395833335,0.845073687498015)},anchor=south west,draw=black,fill=white,align=left}]
\addplot [
color=blue,
solid,
mark=square,
mark options={solid}
]
coordinates{
 (20,22.1123355019221)(21,34.7315673119051)(22,57.6869585587676)(23,101.916461765326)(24,193.046180639532)(25,395.286561103208)(26,874.006968977041)(27,2135.17922496543)(28,5709.11346365603)(29,17155.6046160323)(30,58529.8688705439)(31,225566.91803441)(32,1011658.45992802) 
};
\addlegendentry{Conventional Approach Efficiency $\xi_2$};

\addplot [
color=black,
solid,
mark=triangle,
mark options={solid}
]
coordinates{
 (20,64.3812155161147)(21,108.105392504099)(22,192.292624697582)(23,365.000291520003)(24,742.507495371157)(25,1632.90507532687)(26,3910.36695109634)(27,10274.9842681645)(28,29861.2566611224)(29,97139.0404606149)(30,357427.782111619)(31,1501675.63282193)(32,7323227.18732722) 
};
\addlegendentry{Improved Approach Efficiency $\xi_1$};

\addplot [
color=blue,
solid,
mark=square,
mark options={solid},
forget plot
]
coordinates{
 (20,5.75415408925354)(21,8.5946029728422)(22,13.4130187226186)(23,21.6060769836368)(24,37.5184959315276)(25,68.4357972983852)(26,137.756478662379)(27,297.026424335674)(28,679.451032416149)(29,1803.54365001582)(30,5509.73136040619)(31,17163.8160805742)(32,72732.8526442533) 
};
\addplot [
color=black,
solid,
mark=triangle,
mark options={solid},
forget plot
]
coordinates{
 (20,61.0726585591887)(21,103.242424525924)(22,184.668839613856)(23,352.242801344746)(24,719.844422623456)(25,1590.20400219048)(26,3817.33696415979)(27,10064.0172073156)(28,29300.8268896257)(29,95630.414286776)(30,352347.610390586)(31,1486048.18115505)(32,7249627.69877635) 
};
\draw[thick,->](axis cs: 25.2,5)--(axis cs: 24,742.507495371157);
\draw[thick,->](axis cs: 25.2,5)--(axis cs: 26,874.006968977041);
\draw[thick,->](axis cs: 29,7)--(axis cs: 28,29861.2566611224);
\draw[thick,->](axis cs: 29,7)--(axis cs: 30,5509.73136040619);
\end{semilogyaxis}
\end{tikzpicture}%
\vspace{-3mm}
\caption{Efficiency measures $\xi_1$ and $\xi_2$ for the sum of $N$ independent Weibull RVs. First case $N=2$: the scale parameters are $\beta_i=1,i=1,2.$ and the shape parameters are $k_1=0.4$ and $k_2=0.8$. Second case $N=4$: the scale parameters are $\beta_i=1,i=1,2,3,4$, and the shape parameters are $k_1=0.4$ and $k_2=k_3=k_4=0.8$.}
\label{fig3}
\end{figure}
In Fig. \ref{fig3}, the two efficiency measures $\xi_1$ and $\xi_2$ are plotted as function of the threshold for both cases $N=2$ and $N=4$. We note from this figure that in both cases $\xi_1$ and $\xi_2$ are much more bigger than $1$ and hence the improved and the conventional IS techniques are more efficient than the naive MC simulation technique. Moreover, it is important to point out that the larger is $\gamma_{th}$, the more efficient are the two IS techniques compared to the naive MC simulation. For instance, for a fixed requirement, the naive MC simulation needs approximately $M_{MC}=10^5\times M_C=10^7\times M_I$ simulations runs when $\gamma_{th}=32$ dB and $N=4$, (see Fig. 3). Furthermore, the most interesting result is that our improved IS technique is more efficient than the conventional one for both scenarios $N=2$ and $N=4$. For instance, for $N=4$ and $\gamma_{th}$=32 dB, the conventional IS approach requires approximately $M_C=\frac{\xi_1}{\xi_2}\times  M_I=100\times M_I$ simulation runs to achieve the same variance (accuracy) given by our improved IS technique. In addition, the efficiency of our improved approach compared to the conventional one $\frac{\xi_1}{\xi_2}$ is clearly higher in the case $N=4$ than the one obtained when $N=2$. This is actually expected from the analysis of the variance in Fig. 2.
\par}
\vspace{-3mm}
\bibliography{References}
\bibliographystyle{IEEEtran}
\end{document}